%% file: main.tex
\documentclass[conference]{IEEEtran}
\IEEEoverridecommandlockouts
\usepackage{cite}
\usepackage{amsmath,amssymb,amsfonts}
\usepackage{wrapfig}
\usepackage{algorithmic}
\usepackage{color}
\usepackage[pdftex]{graphicx}
\usepackage{textcomp}
\usepackage{boldline}
\ifCLASSOPTIONcompsoc
\usepackage[caption=false,font=normalsize,labelfont=sf,textfont=sf]{subfig}
\else
\usepackage[caption=false,font=footnotesize]{subfig}
\fi
\usepackage{fixltx2e}
\usepackage{url}

\newcommand{\calX}{\mathcal{X}}
\newcommand{\calY}{\mathcal{Y}}

\input{comments}

\input{macros}
\newcommand{\hiddenproof}[1]{}

\begin{document}

\title{On the Anonymization of Differentially Private Location Obfuscation\\
\thanks{
This work was partially supported by JSPS KAKENHI Grant Number JP17K12667, JP16K16069, and by JSPS and Inria under the project LOGIS of the Japan-France AYAME Program.
}
}

\author{\IEEEauthorblockN{Yusuke Kawamoto}
\IEEEauthorblockA{\textit{National Institute of Advanced}\\ \textit{Industrial Science and Technology (AIST)} \\
Tsukuba, Japan \\
~}
\and
\IEEEauthorblockN{Takao Murakami}
\IEEEauthorblockA{\textit{National Institute of Advanced}\\ \textit{Industrial Science and Technology (AIST)} \\
Tokyo, Japan}
}

\maketitle

\begin{abstract}
Obfuscation techniques in location-based services (LBSs) have been shown useful to hide the concrete locations of service users, whereas they do not necessarily provide the anonymity.
We quantify the anonymity of the location data obfuscated by the planar Laplacian mechanism 
and that by the optimal geo-indistinguishable mechanism 
of Bordenabe et al. 
We empirically show that the latter provides stronger anonymity than the former 
in the sense that more users in the database satisfy $k$-anonymity.
To formalize and analyze such approximate anonymity we introduce the notion of asymptotic anonymity.
Then we show that the location data obfuscated by the optimal geo-indistinguishable mechanism can be anonymized by removing a smaller number of users from the database.
Furthermore, we demonstrate that the optimal geo-indistinguishable mechanism has better utility both for users and for data analysts.
\end{abstract}


\pagestyle{plain}
\thispagestyle{plain}

\allowdisplaybreaks[1]

\section{Introduction}
\label{sec:intro}

Location-based services (LBSs) have been increasingly employed in a variety of applications, including navigation, resource-trucking, 
recommendation, advertising, games, and authentication.
%
One of the popular applications has been to discover interesting locations from collected location data and provide them for third parties.
When the providers of LBSs publish some geographic locations 
of users, 
the accurate locations may reveal private information, 
such as home addresses, health conditions, and political orientation.

To prevent or mitigate the privacy breach, many location obfuscation techniques have been proposed to hide accurate locations of users while providing their approximate information used in LBSs.
For example, the \emph{dummy location insertion}~\cite{Kido:05:ICDE} generates $k-1$ dummy points and makes a user's location indistinguishable among a set of $k$ locations, which provides \emph{$k$-anonymity}.
The \emph{spacial cloaking technique}~\cite{Gruteser:03:MobiSys} chooses a sufficiently large region that includes $k$ indistinguishable locations to achieve $k$-anonymity.
The \emph{location perturbation technique}~\cite{Machanavajjhala:08:ICDE} adds to each location a controlled random noise and guarantees \emph{differential privacy}, independently of any side information that an adversary may possess.

Such perturbation techniques have been developed to construct more practical mechanisms for location obfuscation.
The \emph{planar Laplacian mechanism}~\cite{Andres:13:CCS} satisfies \emph{geo-indistinguishability}, an extended notion of differential privacy to the Euclid distance.
The \emph{optimal geo-indistinguishable mechanism}~\cite{Bordenabe:14:CCS} minimizes the quality loss caused by the perturbation while preserving geo-indistinguishability.

Although these geo-indistinguishable mechanisms hide the concrete locations, 
no prior work has investigated the relationships between geo-indistinguishability and anonymity to our knowledge.
In this paper, we show geo-indistinguishability does not guarantee to provide $k$-anonymity.
This means that the location data obfuscated by geo-indistinguishable mechanisms might be vulnerable to re-identification attacks (e.g.,~\cite{Montazeri:16:ISITA,Murakami:18:ISITA}) for instance when the LBS provider shares the obfuscated data with a malicious data analyst.
Moreover, such leakage of user identity information can be efficiently detected and quantified using an automated tool such as~\cite{chothia2013tool,ChothiaKN14:esorics}.


In this work we empirically explore the relationships among obfuscation, anonymity, and utility for users and for data analysts in geo-indistinguishable location obfuscation.
In particular, we propose a method for effectively anonymizing the obfuscated data by deleting some data before publishing them to third parties.
The overview of the method is shown in Fig.~\ref{fig:overview}.

\begin{figure}[t]
   \centering
   \includegraphics[width=0.99\linewidth]{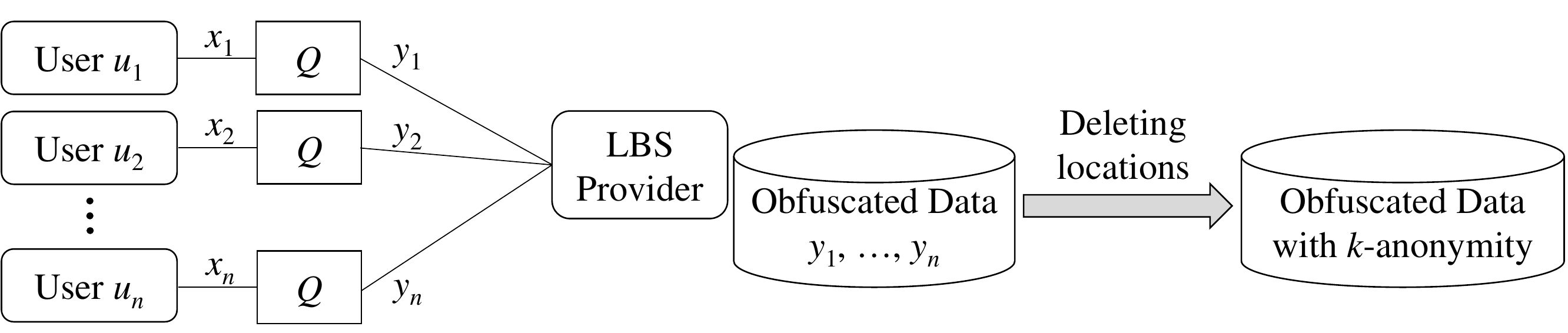}
   \caption{Overview of the proposed method. Each user $u_i$ obfuscates a location $x_i$ using a mechanism $Q$ and sends the obfuscated location $y_i$ to the LBS provider, which anonymizes the collected data to publish them.}
   \label{fig:overview}
\end{figure}

The contributions of this paper are summarized as follows:
\begin{itemize}
\item We evaluate the anonymity of the location data obfuscated by two location obfuscation mechanisms: $\PL$ (the planar Laplacian mechanism) and $\OptQL$ (the optimal geo-indistinguishable mechanism).
We empirically show that $\OptQL$ satisfies stronger anonymity than $\PL$.

\item We propose the notion of \emph{$(\kappa, \alpha)$-asymptotic anonymity}, which generalizes $k$-anonymity to an approximate anonymity of sampled users.

\item We show that the location deletion method, which simply removes the locations of the users who do not satisfy $k$-anonymity, makes the location dataset $k$-anonymous while preserving $\varepsilon$-geo-indistinguishability.
In particular, we demonstrate that $\OptQL$ requires to delete a smaller number of users than $\PL$ to achieve $k$-anonymity.

\item We demonstrate by experiments that the utility for users and for data analysts is better in $\OptQL$ than in $\PL$.
\end{itemize}


\section{Preliminaries}
\label{sec:preliminaries}

For a finite set $\cals$, we denote by $\#\cals$ the number of elements in $S$, and by $\Dists\cals$ the set of all probability distributions over~$\cals$.

\subsection{Obfuscation Mechanism}
\label{sub:notation}

In this work we consider a number $n$ of users each reporting some rough information $y$ on his single geographic location $x$ to an LBS (location-based service) provider while keeping the exact location $x$ hidden from the provider.
To compute an obfuscated location $y$, each user uses a \emph{location obfuscation mechanism} that adds a certain noise to $x$ and outputs it as $y$.

Formally, let $\calX$ be a finite set of all possible locations of the users, and $\calY$ be a finite set of all (possibly fake) locations reported by the users.
Then a \emph{location obfuscation mechanism} (or simply an \emph{obfuscater}) is a probabilistic algorithm $Q: \calX\rightarrow\Dists\calY$ that, given an original location $x$, outputs a \emph{reported location} $y$.
We denote by $Q_{xy}$ the conditional probability that the mechanism $Q$ outputs $y$ given input~$x$.

The probability distribution of the original locations is represented by the prior $\pi$ over $\calX$, and the prior probability of a location $x$ is denoted by $\pi_x$.

\subsection{Geo-indistinguishability}
\label{sub:Geo-IND}

\emph{Geo-indistinguishability}~\cite{Andres:13:CCS} is a notion of location privacy that can be regarded as a variant of local differential privacy \cite{Duchi:13:FOCS} 
in which the privacy budget $\varepsilon$ is multiplied by 
the Euclidean distance $d(x,x')$ between locations $x$ and $x'$. 
\begin{definition}[$\varepsilon$-geo-indistinguishability] \label{def:GI}\rm
Given $\varepsilon \geq 0$, 
an obfuscation mechanism $Q$ provides \emph{$\varepsilon$-geo-indistinguishability} 
if for any inputs $x,x' \in \calX$ and any output $y \in \calY$, we have: 
\begin{align*}
Q_{xy} \leq e^{\varepsilon d(x,x')} Q_{x'y}.
\end{align*}
\end{definition}

Then the difference between $Q_{xy}$ and $Q_{x'y}$ are proportional to the distance between $x$ and $x'$.
%
This implies that geo-indistinguishability 
allows an adversary to infer approximate information about 
the original location (e.g., a user is in Paris), 
but hides the exact location (e.g., home address) from her. 
By relaxing the privacy requirements in this way, 
the amount of noise added to the location can be significantly reduced 
(compared to local differential privacy \cite{Duchi:13:FOCS}). 
Consequently, geo-indistinguishability is useful to implement practical LBSs such as the POI (point of interest) retrieval~\cite{Andres:13:CCS}.

\subsection{Planar Laplacian ($\PL$) Mechanism}
\label{subsec:planar-Laplacian}
The \emph{planar Laplacian ($\PL$) mechanism} \cite{Andres:13:CCS} is an example of the mechanism providing geo-indistinguishability. 
It generates a random noise according to a two-dimensional Laplace distribution, and obfuscates an original location $x$ by adding the noise to~$x$. 
In this paper we use a variant of the planar Laplacian mechanism, which outputs a symbol ``$\bot$'' when the obfuscated location is outside the area of interest $\calX$. 

Formally, the variant planar Laplacian mechanism $\QPL: \calX\rightarrow\Dists(\calX\cup\{\bot\})$ is defined by:
\begin{align*}
\QPL_{xy} = 
\begin{cases}
{\textstyle\frac{1}{c}} \cdot e^{-\varepsilon d(x,y)}
&	\text{(if $y \in \mathcal{X}$)} \\
1 - {\textstyle\frac{1}{c}} \cdot\sum_{y'\in\calX} e^{-\varepsilon d(x,y')}
&	\text{(if $y = \bot$)},
\end{cases}
\end{align*}
where $c = \max_{x}\sum_{y'\in\calX} e^{-\varepsilon d(x,y')}$.
%
Intuitively, $c$ is selected to have the best utility by preventing unnecessarily frequent outputs of $\bot$.

%

\begin{proposition} \label{prop:PL-geo-ind}
$Q^{\PL}$ satisfies $\varepsilon$-geo-indistinguishability.
\end{proposition}

\begin{proof}
$Q^{\PL}$ can be seen as a cascade of the standard planar Laplacian (that does not output $\bot$) and the post-processing algorithm that maps each $y\not\in\calX$ to $\bot$.
It is easy to see that by the triangle inequality, 
the standard planar Laplacian satisfies $\varepsilon$-geo-indistinguishability.
Since differential privacy is immune to post-processing,
$Q^{\PL}$ provides $\varepsilon$-geo-indistinguishability. 
\end{proof}

\subsection{Optimal Geo-indistinguishability ($\OptQL$) Mechanism}
\label{sub:optimal_Geo-IND}

The planar Laplacian mechanism is efficiently computable while the utility of the reported location may not be optimal.
For this reason, Bordenabe \textit{et al.}~\cite{Bordenabe:14:CCS} propose an \emph{optimal geo-indistinguishable location obfuscation mechanism} $\OptQL$ that given a privacy budget $\varepsilon$, minimizes the quality loss (QL) that is defined as the expected value of the Euclidean distance, i.e.,
\begin{align*}
QL(\pi, Q, d) = {\textstyle\sum_{x,y}} \, \pi_{x}Q_{x y} d(x,y)
{.}
\end{align*}

The mechanism $\OptQL$ can be obtained by solving a linear optimization problem that minimizes $QL(\pi, Q, d)$ while satisfying $\varepsilon$-geo-indistinguishability.
However, the computational complexity of this optimization is in $O(\#\calX^3)$.
To reduce this to $O(\#\calX^2)$, they show an approximation technique based on a spanning graph of the set of locations.
See~\cite{Bordenabe:14:CCS} for details.

\subsection{$k$-Anonymity}
\label{subsec:k-anonymity}

The notion of \emph{$k$-anonymity}~\cite{Sweeney:02:IJUFKSa} of a user ensures that the user cannot be distinguished from at least $k-1$ other users being at the same location.
More formally, 
for a positive integer $k$, we say that the users at a location $y$ are $k$-anonymous if $n(y) \geq k$
where $n(y)$ is the number of the users who report $y$ as their locations.
%
We also say that a dataset of locations 
satisfies $k$-anonymity if for every location $y$ in the dataset, the users at $y$ are \emph{$k$-anonymous}.
In this definition $k$-anonymity depends only on the users that have the lowest level of anonymity, and does not take the other users into account.

\section{Anonymization of Obfuscated Location Data}
\label{sec:anonymization}

In this section we address some limitations in the definition of $k$-anonymity and introduce two anonymity notions that generalize $k$-anonymity. 
The first notion measures an obfuscater's capability of anonymization independently of the number $n$ of sampled users.
The second notion extends the first one to take into account the fact that different users in the dataset may have different levels of anonymity. 
Finally, we present a simple solution for enhancing the anonymity of the obfuscated data while preserving geo-indistinguishability.

\subsection{Limitations in the Definition of $k$-Anonymity}
\label{subsec:limitation-k-anonymity}

$k$-anonymity is not always useful to evaluate the level of anonymity in the presence of sampled users.

First, $k$-anonymity in the context of location privacy depends on the number $n$ of the LBS's users in a sample data, and does not solely express an obfuscater $Q$'s capability of anonymization.
For instance, if the number $n$ of sampled users increases then $k$-anonymity tends to hold for a larger value of $k$ (roughly proportionally to $n$) for the same $\pi$ and $Q$.
In other words, $k$-anonymity is not defined as a property of $(\pi, Q)$ independently of the number $n$ of sampled users.

Second, different users in the dataset may have different anonymity levels, 
whereas $k$-anonymity of the dataset depends only on the users that have the lowest level of anonymity.
Hence $k$-anonymity is not expressive enough to 
take into account the different anonymity levels of the other users.

\subsection{$\kappa$-Asymptotic Anonymity}
\label{subsec:asymptotic-anonymity}

To overcome the first limitation described in Section~\ref{subsec:limitation-k-anonymity}, we introduce a notion that expresses an obfuscater $Q$'s capability of anonymization independently of the number $n$ of sampled users.
Intuitively, for a $\kappa\in[0,1]$, we define the notion of  \emph{$\kappa$-asymptotic anonymity} as an extension of $k$-anonymity 
where for any sufficiently large number $n$ of users,
each user is indistinguishable from roughly $n\cdot\kappa - 1$ other users.

Formally, this notion is defined using the probability $p(y)$ that the obfuscation mechanism $Q$ outputs $y$ as follows.

\begin{definition}[$\kappa$-asymptotic anonymity]
\label{def:asymptotic-anonymity}\rm
Given a threshold $\kappa\in[0,1]$, the users at a location $y$ are \emph{$\kappa$-asymptotically anonymous} if $p(y) > \kappa$.
Given a prior $\pi\in\Dists\calX$ and an obfuscater $Q:\calX\rightarrow\Dists\calY$, we say that $(\pi, Q)$ provides \emph{$\kappa$-asymptotic anonymity} if for all $y\in\calY$, $p(y) > 0$ implies $p(y) > \kappa$,
where $p(y) = \sum_{x} \pi_{x} Q_{x y}$.
\end{definition}

Note that $\kappa$ itself can be computed from $\pi$ and $Q$ independently of $n$.
When $(\pi,Q)$ provides $\kappa$-asymptotic anonymity, the number of users required to achieve $k$-anonymity is roughly given by $\frac{k}{\kappa}$.
%
%

\begin{example}[Anonymity of the prior and posterior]\label{eg:prior-posterior-asymptotic}
Let us formalize the asymptotic anonymity before/after applying a mechanism $Q$.
The prior $\pi$ provides ($\min_x \pi_{x}$)-asymptotic anonymity while $(\pi, Q)$ provides ($\min_y \sum_x \pi_{x} Q_{x y}$)-asymptotic anonymity\footnote{Remarkably, the asymptotic anonymity contrasts with the Bayes-vulnerability (aka. converse of the Bayes risk~\cite{Chatzikokolakis:08:JCS}) in quantitative information flow.
Instead of minimization, the prior/posterior Bayes-vulnerabilities are respectively $\max_x \pi_{x}$ and $\sum_y \max_x \pi_{x} Q_{x y}$, and represent the probabilities of an adversary's correctly guessing $x$ in one attempt before/after observing~$y$.
}.
To achieve $k$-anonymity before (resp. after) applying~$Q$, the number of users should be roughly $\frac{k}{\min_x \pi_{x}}$ (resp. $\frac{k}{\min_y \sum_x \pi_{x} Q_{x y}}$).
\end{example}

For a large number $n$ of users, we can compute an approximate maximum value of $\kappa$ from the sample by $\min_y \hat{p}(y) = \min_y \frac{n(y)}{n}$, which converges to $\kappa$ quickly as shown in Fig.~\ref{fig:res_kappa}.

As we will see in Section~\ref{sub:k-anonymity-prior}, $\kappa$-asymptotic anonymity (resp. $k$-anonymity) holds only for small values of $\kappa$ (resp. $k$).
This implies that the obfuscation mechanism does not necessarily provide anonymity to all users although it hides the exact original locations in terms of geo-indistinguishability.

\subsection{$(\kappa, \alpha)$-Asymptotic Anonymity}
\label{subsec:asymptotic-anonymity-delta}
Similarly to $k$-anonymity, the definition of $\kappa$-asymptotic anonymity also suffers from the second limitation described in Section~\ref{subsec:limitation-k-anonymity}.
To evaluate the different levels of anonymity of different users,
we introduce another notion that relaxes $\kappa$-anonymity by allowing some rate $\alpha$ of errors.
Roughly speaking, the new notion expresses that given a sample data with $n$ users, at least $n(1-\alpha)$ users are $n\kappa$-anonymous.

\begin{definition}[$(\kappa,\alpha)$-asymptotic anonymity]
\label{def:asymptotic-delta-anonymity}\rm
Let $p(y) \eqdef \sum_{x} \pi_{x} Q_{x y}$.
Given a $\kappa\in[0,1]$ and an acceptable error rate $\alpha\in[0,1]$, $(\pi,Q)$ provides \emph{$(\kappa, \alpha)$-asymptotic anonymity} if
\[
\frac{ \sum_{y: p(y) > \kappa} p(y) }
     { \sum_{y: p(y) > 0} p(y) }
\geq 1-\alpha
{.}
\]
\end{definition}


This notion can be used to roughly estimate the utility loss in anonymizing the location data.
When there are $n$ users in the dataset, at most $n\alpha$ users are not $n\kappa$-anonymous.
If we remove the locations data of these users, then the dataset will satisfy $n\kappa$-anonymity while the utility of the dataset deteriorates proportionally to the number $n\alpha$ of deleted users.

\subsection{Location Deletion Method ($\Del$) for $k$-Anonymity}
\label{subsec:post-processing}

As explained so far, $\varepsilon$-geo-indistinguishable mechanisms are useful to hide the exact locations from the LBS provider, whereas they may not be able to provide $k$-anonymity of the obfuscated location data.
When the LBS provider wishes to publish such obfuscated data to third parties, a simple solution to achieve $k$-anonymity is what we call the \emph{location deletion method $\Del$}, i.e., to delete the obfuscated locations that do not satisfy $k$-anonymity.
Then the modified database satisfies $k$-anonymity while preserving $\varepsilon$-geo-indistinguishability thanks to the immunity to the post-processing.

More specifically, given a threshold $\kappa$, the minimum number of users that should be removed is approximately given by:
\begin{align*}
n\alpha_{\min} =
n \cdot\biggl(
\frac{ \sum_{y: 0 < p(y) < \kappa} p(y) }
     { \sum_{y: p(y) > 0} p(y) }
\biggr){,}
\end{align*}
where $p(y) = \sum_{x} \pi_{x} Q_{x y}$.
When $Q$ is a Laplacian mechanism, then all locations occur with non-zero probabilities,
and thus the approximate number of deleted users is 
$
n\cdot\sum_{y: p(y) < \kappa} p(y)
$.
We will demonstrate the effect of this combination of obfuscation and anonymization by experiments in Section~\ref{sec:exp}.

\section{Experimental Evaluation}
\label{sec:exp}

In this section we empirically compare the two obfuscation mechanisms $\PL$ and $\OptQL$, and illustrate how the location deletion method $\Del$ enhances the anonymity of obfuscated data 
and affects the utility for users and for data analysts.

\subsection{Experimental Set-up}
\label{sub:setup}
We performed experiments using the Foursquare dataset (Global-scale Check-in Dataset)~\cite{Yang_TIST15}. 
This dataset includes $33278683$ location check-ins by $266909$ users all over the world. 
In our experiments, we used the data in Manhattan, which consists of  
location check-ins by $14951$ users. 
We assumed that each user $u_i$ obfuscated a single location $x_i$ using an $\varepsilon$-geo-indistinguishable obfuscation mechanism $Q$, 
and sent the obfuscated location $y_i$ to the LBS provider. 

We divided Manhattan into $20 \times 20$ regions with regular intervals.
Let $\calX$ be the set of these regions, and $\pi$ be the empirical distribution of the $14951$ users' locations over $\calX$.
We defined the distance $d(x,x')$ between two regions $x$ and $x'$ 
by the Euclidean distance between their central points.
Here we normalized the distance so that the distance between two adjacent regions is one. 

As an obfuscation mechanism $Q$, we employed 
the planar Laplacian mechanism $\PL$ 
(in Section~\ref{subsec:planar-Laplacian})
and
the Optimal geo-indistinguishable mechanism $\OptQL$
(in Section~\ref{sub:optimal_Geo-IND}).
In $\OptQL$, 
we solved the optimization problem\footnote{In $\OptQL$ we set the dilation factor to be $\delta = 1.09$. 
} that minimizes QL while satisfying $\varepsilon$-geo-indistinguishability using 
the linear programming solver \linprog{} in MATLAB. 
For both $\PL$ and $\OptQL$, we set the privacy budget $\varepsilon$ to be 
$0.1$ to $1$, which have been widely used in the literature~\cite{Hsu:14:CSF}. 

After obtaining all obfuscated regions $y_1, y_2, \ldots, y_n$, 
we applied the location deletion method $\Del$ to remove the regions that do not satisfy $k$-anonymity (where $k$ is $10$ or $100$).
We denote by $\PLdel$ (resp.~$\OptQLdel$) 
the application of $\PL$ (resp.~$\OptQL$) post-processed by $\Del$.

\subsection{Experimental Results}
\label{sub:results}

We show the experimental results on anonymity and utility.

\begin{figure}[t]
   \centering
   \includegraphics[width=0.99\linewidth]{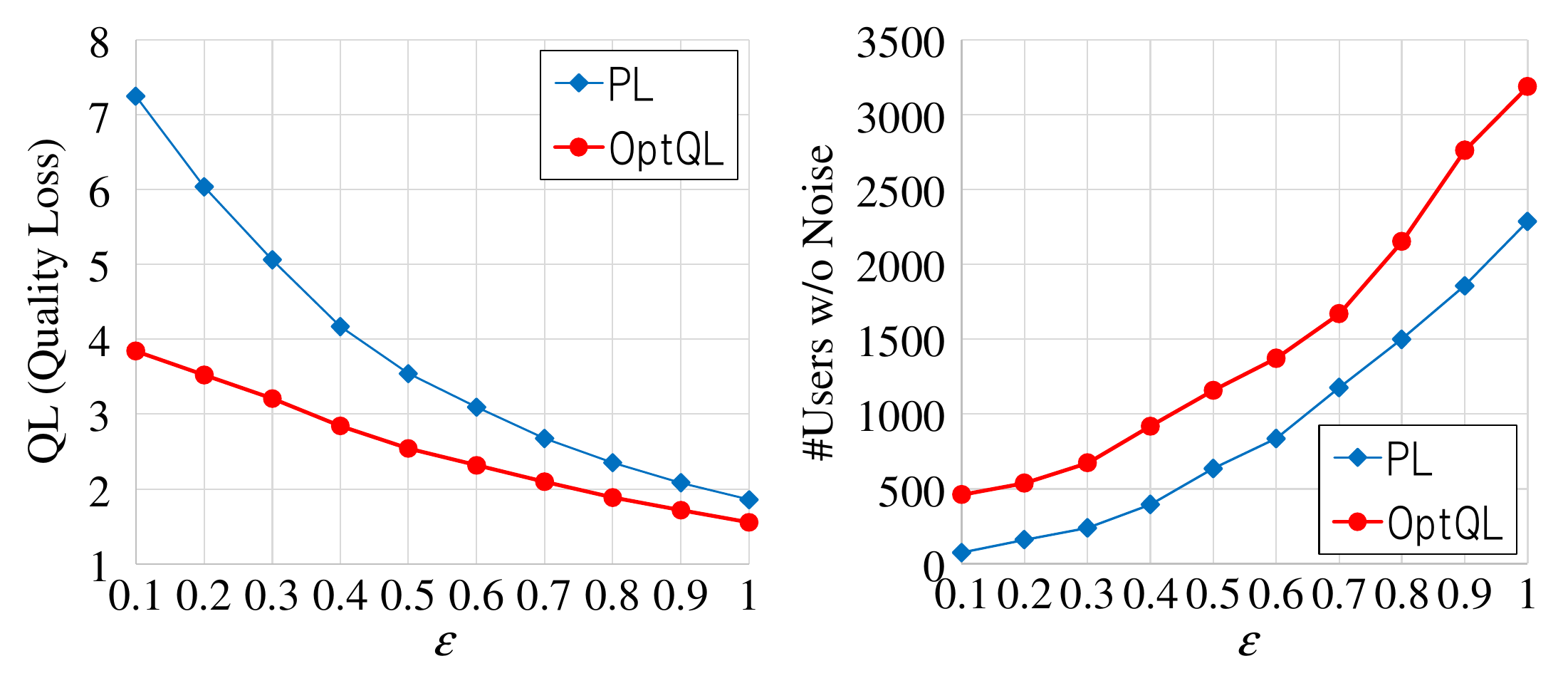}
   \caption{
   Trade-offs between the privacy budget $\varepsilon$ and the utility for users.
   As a utility the graph on the left uses QL (quality loss),
   and the graph on the right uses the number of users that remain at the same regions after obfuscation.
   As for $\PL$, we excluded the users who report $\bot$ as their location.
   }
   \label{fig:res_utility}
\end{figure}

\subsubsection{$k$-anonymity before anonymization}
\label{sub:k-anonymity-prior}
By experiments we found that unless we add much noise, the obfuscation does not provide $k$-anonymity, i.e., $k = 1$ for a user.
Specifically, $k=1$ is provided by $\PL$ for $\varepsilon \geq 0.4$ and by $\OptQL$ for $\varepsilon \geq 0.2$.

\subsubsection{Utility for users}
In Fig.~\ref{fig:res_utility} we compare $\OptQL$ with $\PL$ in terms of the utility for users.
Specifically, 
we evaluated the quality loss, i.e., the average Euclidean distance $d(x_i,y_i)$ between the original region $x_i$ and the obfuscated region $y_i$.
We also evaluated the number of users who remain at the same region after obfuscation, i.e., $x_i = y_i$. 
As shown in Fig.~\ref{fig:res_utility}, for a larger $\varepsilon$, smaller noise is added, hence both $\PL$ and $\OptQL$ have better utility for users;
They decrease the quality loss, and increase the number of users remaining at the same regions. 
The results also demonstrate that $\OptQL$ outperforms $\PL$ in terms of the utility for users.
This is because $\OptQL$ chooses locations that minimize the expected distance, which also makes more users remain at the same regions.

\begin{figure}[t]
  \centering
   \includegraphics[width=0.99\linewidth]{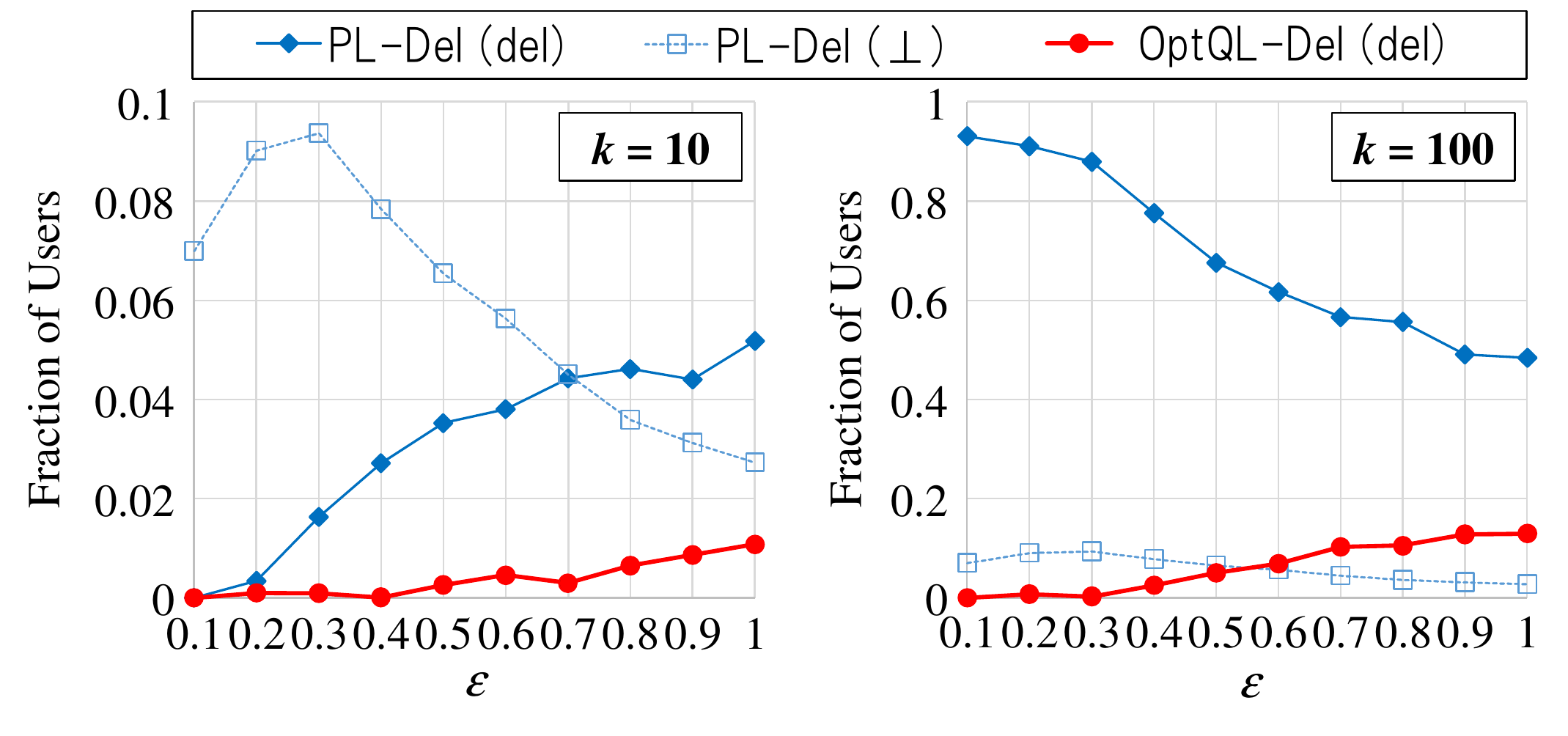}%
  \vspace{-2mm}
  \caption{Trade-offs between the privacy budget $\varepsilon$ and the utility for data analysts.
   The y-axis represents the fraction of deleted users necessary to satisfy $\kappa$-anonymity ($\PLdel$ ($\mathsf{del}$), $\OptQLdel$ ($\mathsf{del}$)), and that of users who output $\bot$ as reported locations ($\PLdel$ ($\bot$)), where $\kappa=6.689\times 10^{-4}$ ($k=10$) on the left and $\kappa=6.689\times 10^{-3}$ ($k=100$) on the right.}
   \label{fig:res_k-anonymity}
\end{figure}

\begin{figure*}[t]
\vspace{-2ex}
  \centering
  \subfloat[When the users reported their original locations.\label{fig:map:original}]{\includegraphics[width=0.32\linewidth]{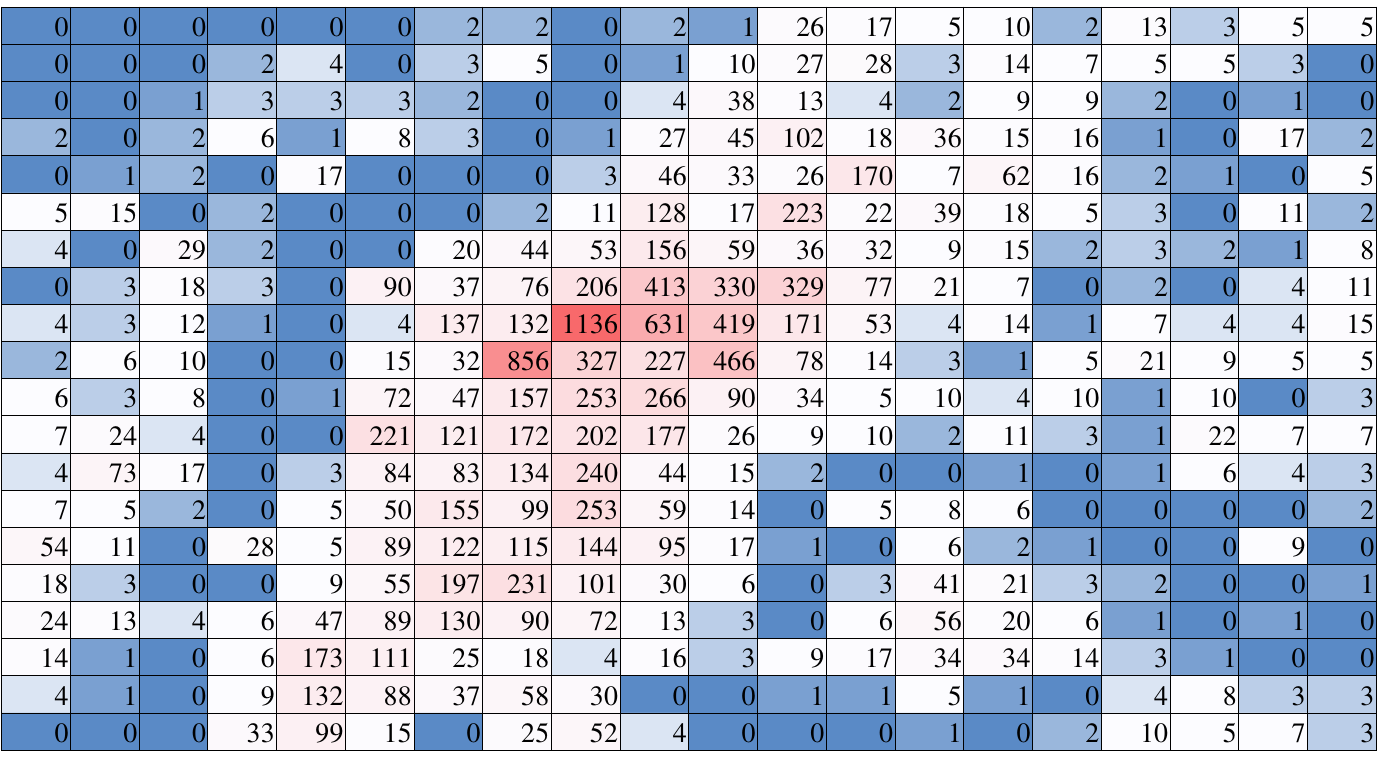}}~~
  \subfloat[When the users reported the locations obfuscated by $\PL$ ($\varepsilon = 1$).\label{fig:map:PL}]{\includegraphics[width=0.32\linewidth]{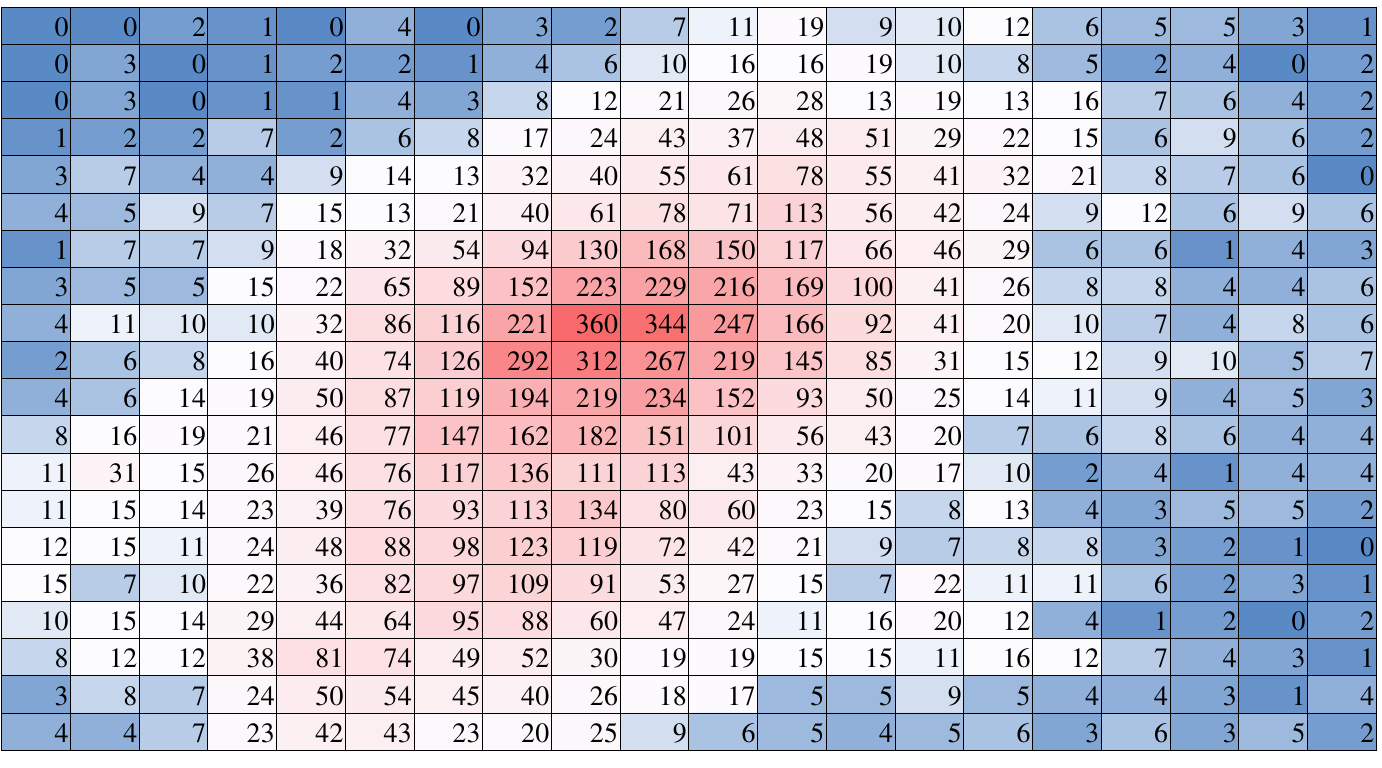}}~~
  \subfloat[When the users reported the locations obfuscated by $\OptQL$ ($\varepsilon = 1$).\label{fig:map:OptQL}]{\includegraphics[width=0.32\linewidth]{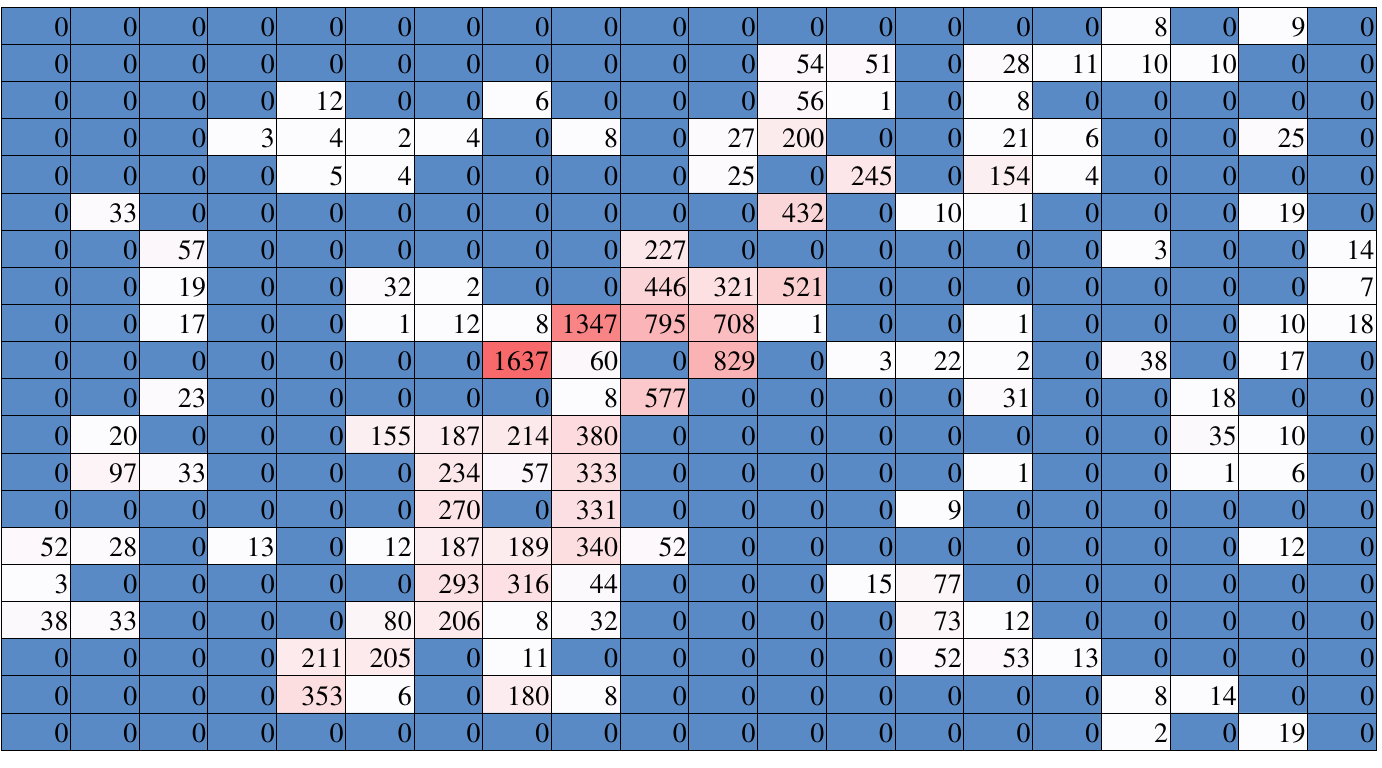}}
%
   \caption{The maps of Manhattan that plot
the numbers of users having reported the regions as their (original/obfuscated) locations.
}
   \label{fig:res_histogram}
\end{figure*}

\subsubsection{Utility for data analysts}
In Fig.~\ref{fig:res_k-anonymity} we compare $\OptQL$ with $\PL$ in terms of the utility for data analysts.
The graphs show the ratio of deleted users for $\kappa=6.689\times 10^{-4}$ ($k = 10$) on the left and for $\kappa=6.689\times 10^{-3}$ ($k = 100$) on the right.
As for $\PL$ we also show the ratio of users reporting $\bot$ as obfuscated regions (indicated as $\PLdel$ ($\bot$)).

According to these graphs,
the ratio of deleted users is significantly smaller in $\OptQLdel$ than in $\PLdel$.
To see this in detail, 
we present the maps of Manhattan that plot the density of the user locations without noise (Fig.~\ref{fig:map:original}), and of those obfuscated by $\PL$ (Fig.~\ref{fig:map:PL}) and by $\OptQL$ (Fig.~\ref{fig:map:OptQL}).

In Fig~\ref{fig:map:PL} we see that the planar Laplacian $\PL$ spreads the population over the whole map.
This is because $\PL$ uniformly draws an angle (from $[0, 2\pi)$) to which it maps each location.
For $\varepsilon\approx0$, the reported regions are distributed almost uniformly.
Hence for a small value of $k$, only a few obfuscated regions need to be deleted to achieve $k$-anonymity (Fig.~\ref{fig:res_k-anonymity} on the left), whereas for a large value of $k$, most of the obfuscated locations need to be deleted (Fig.~\ref{fig:res_k-anonymity} on the right).

In contrast to $\PL$, \textit{the optimal geo-indistinguishable mechanism $\OptQL$ concentrates more users in the crowded regions}  as shown in Fig.~\ref{fig:map:OptQL}.
To see this in detail, we note that for a more crowded region $x$, the prior probability $\pi_x$ is larger.
Since $\OptQL$ tries to minimize $\sum_{x,y} \pi_{x}Q_{x y} d(x,y)$, if $\pi_x$ is larger then $\OptQL$ chooses a region $y$ with a smaller $d(x,y)$, i.e., closer to $x$.
Hence the users located in the crowded regions tend not to move by the obfuscation. 
Conversely, the users outside the crowded regions tend to move to one of the closest crowded regions that provide geo-indistinguishability.

Owing to this concentration, $\OptQL$ provides $(\kappa, \alpha)$-asymptotic anonymity with a smaller error rate $\alpha$. 
For instance, in $\OptQL$, only $161$ users do not satisfy $10$-anonymity ($\alpha = 0.011$), whereas in $\PL$, $773$ users do not ($\alpha = 0.052$).
This means that $\OptQLdel$ removes a smaller number of users than $\PLdel$, and thus has a better utility for data analysts.

To sum up $\OptQLdel$ is more effective than $\PLdel$ in terms of the utility both for users and for data analysts
while providing $\varepsilon$-geo-indistinguishability and $k$-anonymity.

\subsubsection{Convergence of the empirical value of $\kappa$}
\label{sub:anonymity-more-users}
\begin{figure}[t]
  \centering
   \includegraphics[width=0.99\linewidth]{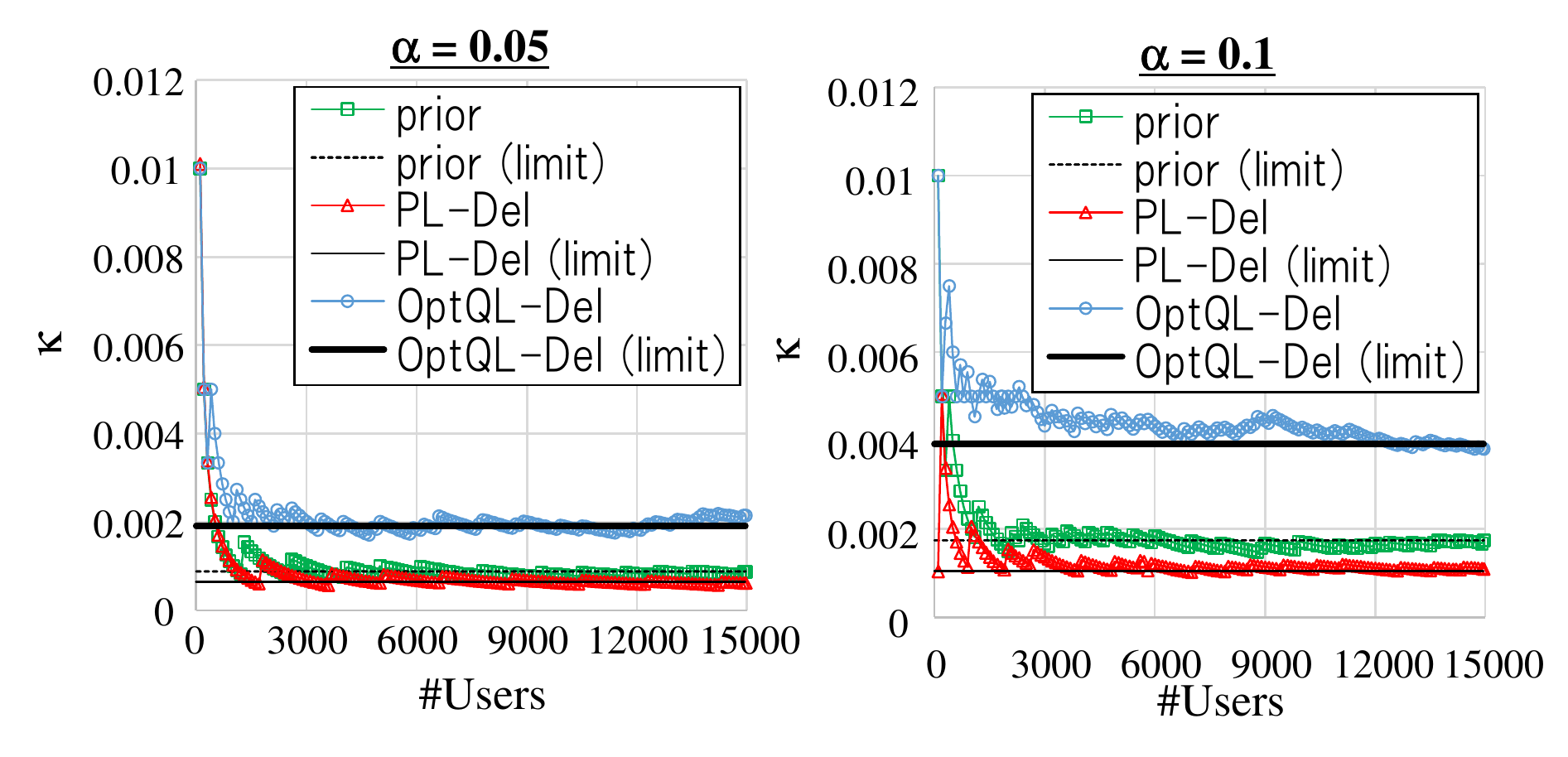}%
  \vspace{-2mm}
  \caption{Relationship between the number of users and the level $\kappa$ of the asymptotic anonymity for $\alpha=0.05$ on the left and for $\alpha=0.1$ on the right. As for $\PLdel$ we excluded $\bot$ from the computation of $\kappa$ for clarity.}
   \label{fig:res_kappa}
\end{figure}

\begin{table}[t]
\centering
\caption{The level $\kappa$ of asymptotic anonymity converges to the following values when increasing the number of users. \label{table:kappa:theory}}
\begin{tabular}{lcc}
\hline
& $\alpha=0.05$ & $\alpha=0.1$ \\ \hlineB{2}
prior (without noise) 
& $8.7\times 10^{-4}$ & $1.7\times 10^{-3}$ \\ \hline
after applying $\PLdel$ 
& $6.4\times 10^{-4}$ & $1.0\times 10^{-3}$ \\ \hline
after applying $\OptQLdel$ 
& $1.9\times 10^{-3}$ & $3.9\times 10^{-3}$ \\ \hline \\[-2ex]
\end{tabular}
\end{table}

In Fig.~\ref{fig:res_kappa} we show how the empirically computed value of $\kappa$ converges to the value displayed in Table~\ref{table:kappa:theory} when increasing the number $n'$ of users.
In the experiments we uniformly sampled a subset (of size $n'$) from the original dataset, applied each mechanism with $\varepsilon = 1$, and computed the maximum $\kappa$ such that $n'(1-\alpha)$ users satisfy $n'\kappa$-anonymity (for $\alpha = 0.05,\, 0.1$).
These graphs imply that $\kappa$ is (roughly) independent of $n'$ and thus $\kappa$-asymptotic anonymity can be seen as a property of the prior and obfuscater.
Therefore $\kappa$ is useful to learn that given a different number $n$ of sampled users, the dataset roughly satisfies $n\kappa$-anonymity.

\section{Conclusion}
\label{sec:conclusion}
We have empirically evaluated the anonymity of the location data obfuscated by $\PL$ and by $\OptQL$,
and shown that $\OptQL$ provides stronger anonymity than $\PL$ 
in the sense that it requires to remove a fewer users to achieve $k$-anonymity.
To analyze this formally, 
we have introduced the notion of $(\kappa, \alpha)$-asymptotic anonymity.
We have also demonstrated that $\OptQL$ has better utility for users and for data analysts.

In future work we plan to develop a utility-optimal obfuscater satisfying geo-indistinguishability and anonymity.
We will also explore rigorous foundations of obfuscation based on statistics, and relationships with quantitative information flow.

\bibliographystyle{IEEEtran}
\bibliography{short}

\end{document}
\newpage

\section{Ease of Use}

\subsection{\LaTeX-Specific Advice}

Please use ``soft'' (e.g., \verb|\eqref{Eq}|) cross references instead
of ``hard'' references (e.g., \verb|(1)|). 
%
Use \verb|{align}| or \verb|{IEEEeqnarray}| instead. 
%
%
%
%
%

\subsection{Some Common Mistakes}\label{SCM}
\begin{itemize}
\item The word ``data'' is plural, not singular.
\item In American English, commas, semicolons, periods, question and exclamation marks are located within quotation marks only when a complete thought or name is cited, such as a title or full quotation. When quotation marks are used, instead of a bold or italic typeface, to highlight a word or phrase, punctuation should appear outside of the quotation marks. A parenthetical phrase or statement at the end of a sentence is punctuated outside of the closing parenthesis (like this). (A parenthetical sentence is punctuated within the parentheses.)
\item A graph within a graph is an ``inset'', not an ``insert''. The word alternatively is preferred to the word ``alternately'' (unless you really mean something that alternates).
\item Do not use the word ``essentially'' to mean ``approximately'' or ``effectively''.
\item In your paper title, if the words ``that uses'' can accurately replace the word ``using'', capitalize the ``u''; if not, keep using lower-cased.
\item Be aware of the different meanings of the homophones ``affect'' and ``effect'', ``complement'' and ``compliment'', ``discreet'' and ``discrete'', ``principal'' and ``principle''.
\item Do not confuse ``imply'' and ``infer''.
\item The prefix ``non'' is not a word; it should be joined to the word it modifies, usually without a hyphen.
\item There is no period after the ``et'' in the Latin abbreviation ``et al.''.
\item The abbreviation ``i.e.'' means ``that is'', and the abbreviation ``e.g.'' means ``for example''.
\end{itemize}
An excellent style manual for science writers is \cite{b7}.

\subsection{Figures and Tables}
\paragraph{Positioning Figures and Tables} Place figures and tables at the top and 
bottom of columns. Avoid placing them in the middle of columns. Large 
figures and tables may span across both columns. Figure captions should be 
below the figures; table heads should appear above the tables. Insert 
figures and tables after they are cited in the text. Use the abbreviation 
``Fig.~\ref{fig}'', even at the beginning of a sentence.

\begin{table}[htbp]
\caption{Table Type Styles}
\begin{center}
\begin{tabular}{|c|c|c|c|}
\hline
\textbf{Table}&\multicolumn{3}{|c|}{\textbf{Table Column Head}} \\
\cline{2-4} 
\textbf{Head} & \textbf{\textit{Table column subhead}}& \textbf{\textit{Subhead}}& \textbf{\textit{Subhead}} \\
\hline
copy& More table copy$^{\mathrm{a}}$& &  \\
\hline
\multicolumn{4}{l}{$^{\mathrm{a}}$Sample of a Table footnote.}
\end{tabular}
\label{tab1}
\end{center}
\end{table}
%
%

\section*{Acknowledgment}
Put sponsor acknowledgments in the unnumbered footnote on the first page.

\end{document}

%% file: comments.tex
\newif\ifcommentson\commentsonfalse
\ifcommentson
\newcommand{\commentsize}[0]{.45\textwidth}
\newcommand{\commentTM}[1]{\begin{center} \parbox{\commentsize}{\textbf{\textcolor{black}{Comment T.}} \textcolor{red}{#1 }}\end{center}}
\newcommand{\commentYK}[1]{\begin{center} \parbox{\commentsize}{\textbf{\textcolor{black}{Comment Y.}} \textcolor{red}{#1} }\end{center}}
\newcommand{\replyTM}[1]{\begin{center} \parbox{\commentsize}{\textbf{Reply T.} \textcolor{blue}{#1} }\end{center}}
\newcommand{\replyYK}[1]{\begin{center} \parbox{\commentsize}{\textbf{Reply Y.} \textcolor{blue}{#1} }\end{center}}
\newcommand{\commentT}[1]{\marginpar{\footnotesize \color{red} {\bf T:} \textsf{\scriptsize #1}}}
\newcommand{\commentY}[1]{\marginpar{\footnotesize \color{red} {\bf Y:} \textsf{\scriptsize #1}}}
\newcommand{\replyT}[1]{\marginpar{\footnotesize \color{red} {\bf T:} \textsf{\scriptsize #1}}}
\newcommand{\replyY}[1]{\marginpar{\footnotesize \color{red} {\bf Y:} \textsf{\scriptsize #1}}}
\else
\newcommand{\commentTM}[1]{}
\newcommand{\commentYK}[1]{}
\newcommand{\replyTM}[1]{}
\newcommand{\replyYK}[1]{}
\newcommand{\commentT}[1]{}
\newcommand{\commentY}[1]{}
\newcommand{\replyT}[1]{}
\newcommand{\replyY}[1]{}
\fi

\newcommand{\pagelimitmarker}[1]{
\section*{\textcolor{red}{\ifthenelse{\thepage>#1}{\huge Exceeding the page limit}{\Large Within the page limit}}}
\section*{\textcolor{red}{Page Limit =  #1 ~~~ Current Page = $\thepage$}}}

%% file: macros.tex
\newtheorem{theorem}{Theorem}
\newtheorem{proposition}[theorem]{Proposition}

\newtheorem{definition}{Definition}
\newtheorem{example}{Example}
\newenvironment{proof}{\noindent\textit{Proof: }}{\hspace*{\fill} $\Box$\\}

\newcommand{\eqdef}{\ensuremath{\stackrel{\mathrm{def}}{=}}}

\newcommand{\Dists}{\mathbb{D}} 

\newcommand{\cals}{\mathcal{S}}

\newcommand{\linprog}{\textsf{linprog}}

\newcommand{\QPL}{Q^{\PL}}

\newcommand{\PL}[0]{\mathsf{PL}}
\newcommand{\OptQL}[0]{\mathsf{OptQL}}
\newcommand{\Del}[0]{\mathsf{Del}}
\newcommand{\PLdel}[0]{\PL{}\!\mathsf{-}\Del{}}
\newcommand{\OptQLdel}[0]{\OptQL{}\!\mathsf{-}\Del{}}